\documentclass[journal]{IEEEtran}
\usepackage{stmaryrd}
\usepackage{amsfonts,array}
\usepackage[dvipsnames,usenames]{color}
\usepackage{graphicx,times,psfig,amsmath} 
\usepackage{graphics,color}
\usepackage{graphicx}
\usepackage{amssymb,amsmath,amsfonts} 
\usepackage{amsmath,mathrsfs,bm,url,times}
\usepackage{latexsym}
\usepackage{subfigure}
\usepackage{algorithm}
\usepackage{algorithmicx}
\usepackage{algpseudocode}
\usepackage{epsfig,url}
\usepackage{epstopdf}
\usepackage{multirow}

\def\QED{\hspace*{\fill}~\QEDclosed}
\newtheorem{proposition}{Proposition}
\newtheorem{theorem}{Theorem}
\newtheorem{lemma}{Lemma}

\newtheorem{definition}{Definition}

\newtheorem{remark}{Remark}

\DeclareMathOperator*{\argmax}{argmax}
\newcommand{\R}{{\mathbb R}}
\newcommand{\onetom}{1,2,\cdots,m}
\newcommand{\oneton}{1,\cdots,n}

\newcommand{\zerotoinfty}{0,1,2,\cdots}

\newcommand{\measure}{\text{\it m}}
\newcommand{\ignore}[1]{}

\begin{document}
\title{Global Convergence of Analytic Neural Networks with Event-triggered Synaptic Feedbacks}
\author{Wenlian Lu~\IEEEmembership{Member,~IEEE}, Ren Zheng, Xinlei Yi ,  Tianping Chen~\IEEEmembership{Senior~Member,~IEEE}\thanks{This manuscript is submitted to the Special Issue on Neurodynamic Systems for Optimization and Applications.}\thanks{R. Zheng, X. Yi, W. L. Lu and T. P. Chen are with the School of Mathematical Sciences, Fudan University, China; W. L. Lu is also with the Centre for Computational Systems Biology and School of Mathematical Sciences, Fudan University, and Department of Computer Science, The University of Warwick, Coventry, United Kingdom; T. P. Chen is also with the School of Computer Science, Fudan University, China (email: \{yix11, wenlian, tchen\}@fudan.edu.cn).} \thanks{This work is jointly supported by the Marie Curie International Incoming Fellowship from the European
Commission (FP7-PEOPLE-2011-IIF-302421), the National Natural Sciences Foundation
of China (Nos. 61273211 and 61273309), the Program for New Century Excellent Talents in University (NCET-13-0139), and the Programme of Introducing Talents of Discipline to Universities (B08018).}}



\date{}

\maketitle

\begin{abstract}
In this paper, we investigate convergence of a class of analytic neural
networks with event-triggered rule. This model is general and include Hopfield neural network as a special case. The event-trigger rule efficiently reduces the frequency of information transmission between synapses of the neurons. The synaptic feedback of each neuron keeps a constant value based on the
outputs of its neighbours at its latest triggering time but changes until the next
triggering time of this neuron that is determined by certain criterion via its
neighborhood information. It is proved that the analytic neural
network is completely stable under this event-triggered rule. The main
technique of proof is the ${\L}$ojasiewicz inequality to prove the finiteness of trajectory length. The realization of this event-triggered rule
is verified by the exclusion of Zeno behaviors. Numerical examples are
provided to illustrate the theoretical results and present the optimisation
capability of the network dynamics.
\end{abstract}
\begin{IEEEkeywords}
Analytic neural network, complete stability, distributed event-triggered rule, Self-triggered rule
\end{IEEEkeywords}

\section{Introduction}

\PARstart{T}{his} paper focuses on the following dynamical system
\begin{align}
\begin{cases}
\dot{x}=-D x-\nabla f(y)+\theta\\[3pt]
y=g(\Lambda x),\\[2pt]
\end{cases}\label{mg0.1}
\end{align}
where $x\in\R^n$ is the state vector, $D$ is a constant self-inhibition matrix, the cost function $f(y)$ is an {\em analytic} function and $\theta\in\R^n$ is a constant input vector. $y=g(\Lambda x)$ is the output vector with the sigmoid function $g(\cdot)$ as nonlinear activation function.

Eq. \eqref{mg0.1} was firstly proposed in \cite{Mfa} and is a general model of neural-network system arising in recent years. For example, the well-known Hopfield neural network \cite{Jjh1,Jjh}, whose continuous-time version can be formulated as
\begin{align}
\begin{cases}
C_i\dot{x}_{i}=-\dfrac{x_{i}}{R_{i}}+\sum\limits_{j=1}^{n}\omega_{ij}y_{j}+\theta_{i}\\[3pt]
y_{i}=g_{i}(\lambda_{i}x_{i}),\\[2pt]
\end{cases}\label{mg0.2}
\end{align}
for $i=\oneton$, where $x_{i}$ stands for the state of neuron $i$ and each activation function $g_{i}(\cdot)$ is  sigmoid. With the symmetric weight condition ($\omega_{ij}=\omega_{ji}$ for all $i,j=1,\cdots,n$), Eq. (\ref{mg0.2}) can be formulated as Eq. (\ref{mg0.1}) with $f(y)=-\frac{1}{2}\sum_{i,j=1}^{m}\omega_{ij}y_iy_j$. This model has a great variety of applications. It can be used to search for local minima of the quadratic objective function of $f(y)$ over the discrete set $\{0,1\}^{n}$ \cite{Mvi}-\cite{Wll}, for example, the traveling-sales problem \cite{Jma}.
One step further, this model was extended for a multi-linear cost function $E(y)=\sum_{i_1,\cdots,i_p}a_{i_1,\cdots,i_p}y_{i_1},\cdots,y_{i_p}$ \cite{Mvi}. This model can be also regarded as a special form of (\ref{mg0.1}) with $f(y)=E(y)$ and was proved that this model can minimize $E(y)$ over the discrete set $\{0,1\}^{n}$ \cite{Mvi}.

In application for optimisation, analysis of convergence dynamics is fundamental, which has attracted many interests from different fields. See \cite{Mac}-\cite{Tpcl} and the references therein. The linearization technique and the
classical LaSalle approach for proving stability \cite{Mac,Mvi} could be invalid when the system had non-isolated equilibrium
points (e.g., a manifold of equilibria) \cite{Mfa}. A new concept
"absolute stability" was proposed in \cite{Mfa1,Mfa,Mfa2} to show that
each trajectory of the neural network converges to certain equilibrium for any parameters and
activation functions satisfying certain conditions by proving the finiteness of the
trajectory length and the celebrated ${\L}$ojasiewicz inequality
\cite{Slo}-\cite{Slo1}. This idea can also be seen in an earlier paper \cite{Tpc}.

However, in the model (\ref{mg0.1}), the synaptic feedback of each neuron is continuous bsed on the output states of its neighbours, which is costly in practice for a network of a large number of neurons. In recent years, with the development of sensing, communications, and computing equipment, event-triggered control
\cite{Pta}-\cite{Yfg} and self-triggered control \cite{Aapt}-\cite{Aap}
have been proposed and proved effective in reducing the frequency of synaptic information exchange significantly.
In this paper, we investigate global convergence of analytic
neural networks with event-triggered synaptic feedbacks. Here, we present
event-triggered rules to reduce the frequency of receiving synaptic
feedbacks. At each neuron, the synaptic feedback is a constant that is determined by the outputs of
its neighbours at its latest triggering time and changes at the next triggering time
of this neuron that is triggered by a criterion via its neighborhood
information as well. We prove that the analytic neural networks are
convergent (see {\it Definition \ref{convergence}}) under these event-triggered rules by the ${\L}$ojasiewicz inequality. In addition, we further prove that the event-triggered rule is viable, owing to the
exclusion of Zeno behaviors. These event-triggered rules are distributed
(each neuron only needs the information of its neighbours and itself),
asynchronous, (all the neurons are not required to be triggered in a
synchronous way), and independent of each other (triggering of an neuron
will not affect or be affected by triggering of other neurons). It should be
highlighted that our results can be easily extended to a large class of
neural networks. For example, the standard cellular networks
\cite{Loc}-\cite{Loc1}.

The paper is organized as follows: in Section \ref{sec2}, the preliminaries are given; in Section \ref{sec3}, the convergence and the Zeno behaviours of analytic neural networks with the triggering rules : distributed event-triggered rule is proved in Section \ref{sec3}; in Section \ref{sec5}, examples with numerical simulation are provided to show the effectiveness of the
theoretical results and illustrate its application; the paper is concluded in Section \ref{sec6}.

\noindent {\bf Notions}: $\R^{n}$ denotes $n$-dimensional real space. $\|\cdot\|$ represents the Euclidean norm for vectors or the induced 2-norm for matrices. $B_r(x_0)=\{x\in\R^{n}:~\|x-x_0\|<r\}$ stands for an $n$-dimensional ball
with center $x_0\in\R^{n}$ and radius $r>0$. For a function
$F(x):~\R^{n}\rightarrow\R$, $\nabla F(x)$ is its gradient. For a set $Q\subseteq \R^{n}$ and a point $x_0\in\R^{n}$,
$\text{\it dist\,}(x_0,Q)=\inf_{y\in Q}\|x_0-y\|$ indicates the distance from $x_0$ to
$Q$.

\section{Preliminaries and problem formulation}\label{sec2}

In this section, we firstly provide some definitions and results on algebraic graph theory, which will be used later. (see the textbooks \cite{Rdi}, \cite{RaH} for details)

For a directed graph $\mathcal G=(\mathcal V,\mathcal E,\mathcal A)$ of $n$ neurons (or nodes). where $\mathcal V=\{v_{1},\cdots,v_{n}\}$ is the set of neurons, $\mathcal E\subseteq\mathcal V\times\mathcal V$ is the set of the links (or edges), and $\mathcal A=[a_{ij}]_{n\times n}$ with nonnegative adjacency elements $a_{ij}\in\{0,1\}$ is the adjacency matrix, a link of $\mathcal G$ is denoted by $e(i,j)=(v_{i},v_{j})\in\mathcal E$ if there is a directed link from neuron $v_{j}$ to $v_{i}$ and the adjacency elements associated with the links of the graph
are positive, (i.e., $e(i,j)\in\mathcal E$ if and only if $a_{ij}>0$). We take $a_{ii}=0$ for all $i=\oneton$. Moreover, the in-neighbours and out-neighbours set of neuron $v_{i}$ are defined as $N_{i}^{\text{in}}=\{v_{j}\in \mathcal V:a_{ij}>0\}$ and $N_{i}^{\text{out}}=\{v_{j}\in \mathcal V:a_{ji}>0\}$. The neighbours of the neuron $v_{i}$ denoted by $N_{i}$ is the union of in-neighbours $N_{i}^{\text{in}}$ and out-neighbours $N_{i}^{\text{out}}$, that is, $N_{i}=N_{i}^{\text{in}}\bigcup N_{i}^{\text{out}}$.

Consider the discrete-time synaptic feedback, Eq. \eqref{mg0.1} can be reformulated as follows
\begin{align*}
\begin{cases}
\dot{x}_{i}(t)=-d_{i} x_{i}(t)-\Big[\nabla f\big(y(t^i_{k_i(t)})\big)\Big]_{i}+\theta_{i}\\[5pt]
y_{i}(t)=g_{i}\big(\lambda_ix_i(t)\big)\\[3pt]
\end{cases}
\end{align*}
for $i=\oneton$ and $k_i(t)=\zerotoinfty$, where $x_{i}\in\R$ , $d_{i}>0$ and $\theta_{i}\in\R$.  $f(y):\R^{n}\to\R$ is an {\em analytic} cost function function and $y_{i}=g_{i}(\lambda_{i}x_{i})$ is the output vector with a scaling parameter $\lambda_{i}>0$ and the sigmoid functions $g_{i}(\cdot)$ as nonlinear activation functions. In this paper, we take
\begin{align*}
g_{i}(x)=\frac{1}{1+e^{-x}}
\end{align*}
and the gradient of the activation function $g(\cdot)$ at $x\in\R^n$ can be written as $\partial g(x)=\text{\it diag}\{g'_{1}(x_{1}),\cdots,g'_{n}(x_{n})\}$. The strict increasing triggering event time sequence $\{t_{k}^{i}\}_{k=1}^{+\infty}$ (to be defined) are
neuron-wise and $t_{1}^{i}=0$, for all $i=\oneton$. At each $t$, each neuron $i$ changes the information from its neighbours with respect to an identical time point $t_{k_{i}(t)}^{i}$ with $k_{i}(t)=\argmax_{k'}\{t^{i}_{k'}\leqslant t\}$. Throughout the paper, We may simplify the notation $t_{k_{i}(t)}^{i}$ as $t_{k}^{i}$ unless there is a potential ambiguity. Thus, we have
\begin{align}
\begin{cases}
\dot{x}_i(t)=-d_{i} x_{i}(t)-\Big[\nabla f\big(y(t^i_{k})\big)\Big]_{i}+\theta_{i}\\[5pt]
y_{i}(t)=g_{i}\big(\lambda_ix_i(t)\big)\\[3pt]
\end{cases}\label{mg}
\end{align}
for $i=\oneton$ and $k=\zerotoinfty$. Let $F(x)=[F_{1}(x),\cdots,F_{n}(x)]^{\top}$ be the vector at the right-hand side of \eqref{mg}, where
\begin{align*}
F_{i}(x)=-d_{i} x_{i}-\Big[\nabla f\big(y(t^i_{k})\big)\Big]_{i}+\theta_{i}.
\end{align*}
Note that when we consider the trajectories, the right-hand side of \eqref{mg} can be written as $F(x(t))=[F_{1}(x(t)),\cdots,F_{n}(x(t))]^{\top}$.
Denote the set of equilibrium points for \eqref{mg} as
\begin{align*}
\mathcal S=\Big\{x\in\R^{n}:F(x)=0\Big\}.
\end{align*}
We first recall the definition of convergence for model \eqref{mg} \cite{Mhi}.
\begin{definition}\label{convergence}\cite{Mfa} Given an analytic function $f(\cdot)$, a sigmoid function $g_{i}(\cdot)$ and three constants $d_{i}$, $\theta_{i}$ and $\lambda_{i}$ specifically, system \eqref{mg} is said to be {\it convergent} \footnote{~This definition is frequently referred to as {\it complete stability} of the system in the neural network literature, see \cite{Mfa}.}
if and only if, for any trajectory $x(t)$ of \eqref{mg}, there exists $x^{\bm*}\in\mathcal S$ such that
\begin{align*}
\lim_{t\to+\infty}x(t)=x^{\bm*}.
\end{align*}
\end{definition}

Since the $\omega$-limit set of any trajectory $x(t)$ for the system \eqref{mg} (i.e., the set of points that are approached by $x(t)$ as $t\to+\infty$) is isolated equilibrium points, the convergence of the system \eqref{mg} is global. Our main focus lies in proving that the state $x(t)$ of the system \eqref{mg} under some given rule can converge to these equilibrium points.

The following lemma shows that all solutions for \eqref{mg} are bounded and there exists at least one equilibrium point.

\begin{lemma}\label{Existence}
Given a constant matrix $D$, a constant vector $\theta$ and two specific functions $f(\cdot)$ and $g(\cdot)$, for any triggering event time sequence $\{t^{i}_k\}_{k=0}^{+\infty}~(i=\oneton)$, there exists a unique solution for the piece-wise cauchy problem \eqref{mg} with some initial data $x(0)\in\R^{n}$. Moreover, the solutions with different initial data are bounded for $t\in[0,+\infty)$.
\end{lemma}

\begin{proof} Firstly, we prove the existence and uniqueness of the solution for the system (\ref{mg}). Denotes $t_{k}=[t^{1}_{k},\cdots,t^{n}_{k}]^{\top}$, where $k=\zerotoinfty$. Given a time sequence $\{t^{i}_k\}_{k=0}^{+\infty}~(i=\oneton)$ ordered as $0=t^{i}_{0}<t^{i}_{1}<t^{i}_{2}<\cdots<t^{i}_{k}<\cdots$ (same items in $\{t^{i}_k\}_{k=0}^{+\infty}$ treat as one), there exists a unique solution of \eqref{mg} in the interval $[t_{0},t_{1})$ by using $x(t_{0})=x(0)$ as the initial data (see, existence and uniqueness theorem in \cite{Jkh}). For the next interval $[t_{1},t_{2})$, $x(t_{1})$ can be regarded as the new initial data, which can derive another unique solution in this interval. By induction, we can conclude that there exists a piecewise unique solution over the time interval $t\in[0,+\infty)$, which is for the cauchy problem \eqref{mg} with the initial data $x(0)$.

Secondly, since $0<g_i(x)<1~(x\in\R)$, there exists a constant $M>0$ such that
\begin{align*}
-d_{i}x_{i}(t)-M\leqslant F_{i}\big(x(t)\big)\leqslant-d_{i}x_{i}(t)+M
\end{align*}
Thus for any $\varepsilon_{0}>0$, there exists $r_0>0$ such that
\begin{align*}
\begin{cases}
F_i\big(x(t)\big)<-\varepsilon_0,&~\forall~ x_i(t)\geqslant r_0\\
F_i\big(x(t)\big)>\varepsilon_0,&~\forall~ x_i(t)\leqslant -r_0
\end{cases}
\end{align*}
where $i=\oneton$. Let
\begin{align*}
\mathcal B=\Big\{x\in\R^{n}:\|x\|\leqslant r_{0}\Big\}
\end{align*}
If $x(0)\notin\mathcal B$, $x(t)$ will drop into the set $\mathcal B$ in finite time, which implies that $\mathcal B$ is the $\omega$-limit set of any trajectory $x(t)$ and it is also positively invariant. Thus all the solutions of \eqref{mg} with different initial data are eventually confined in $\mathcal B$, hence they are bounded for $t\in[0,+\infty)$.
\end{proof}

Consider now the set of equilibrium points $\mathcal S$. The following lemma is established in \cite{Mfa}, which shows that there exists at least one equilibrium point in $\mathcal S$.
\begin{lemma}\label{nontrivial}
For the of equilibrium points  $\mathcal S$, the following statements hold:
\begin{enumerate}
\renewcommand{\labelenumi}{\it(\arabic{enumi})}
\item $\mathcal S$ is not empty.\\[-10pt]
\item There exists a constant $r>0$ such that
\begin{align*}
\mathcal S\bigcap\,\Big(\R^{n}\setminus B_{r}(\bf0)\Big)=\emptyset.
\end{align*}
\end{enumerate}
\end{lemma}

To depict the event that triggers the next feedback basing time point, we introduce the following candidate Lyapunov (or energy) function:
\begin{align}
L(x)=&f(y)+\sum_{i=1}^{n}\bigg[\frac{d_{i}}{\lambda_{i}}\int_{0}^{y_{i}}g_{i}^{-1}(s)ds\bigg]
-\theta^{\top} y.\label{Ly}
\end{align}
where $y=[y_{1},\cdots,y_{n}]$ with $y_{i}=g_{i}(\lambda_{i}x_{i})~(i=\oneton)$. The function $L(x)$ generalizes the Lyapunov function introduced for \eqref{mg0.1} in \cite{Mvi}, and it can also be thought of as the energy function for the Hopfield and the cellular neural networks model \cite{Mac,Loc}. In this paper, we will prove that the candidate Lyapunov function \eqref{Ly} is a strict Lyapunov function \cite{Mfa}, as stated in the following definition.
\begin{definition}
A Lyapunov function $L(\cdot):\R^{n}\rightarrow\R$ is said to be strict if $L\in C^{1}(\R^{n})$, and the derivative of $L$ along trajectories $x(t)$, i.e. $\dot{L}(x(t))$, satisfies $\dot{L}(x)\leqslant
0$  and  $\dot{L}(x)< 0$ for $x\notin \mathcal S$. 
\end{definition}

The next lemma provides an inequality, named ${\L}$ojasiewicz inequality \cite{Slo}. It will be used to prove the
finiteness of length for any trajectory $x(t)$ of the system \eqref{mg}, which can finally derive the convergence of system \eqref{mg}. The definition of trajectory length is also listed in Definition \ref{def2}.
\begin{lemma}\label{Loj}
Consider an analytic and continuous function $H(x):\mathcal D\subseteq\R^{n}\rightarrow\R$. Let
\begin{align*}
\mathcal S_{\nabla}=\Big\{x\in\mathcal D:\nabla H(x)=\bf 0\Big\}.
\end{align*}
For any $x_{s}\in\mathcal S_{\nabla}$, there exist two constants $r(x_{s})>0$ and $0<v(x_{s})<1$, such that
\begin{align*}
\big|H(x)-H(x_{s})\big|^{v(x_{s})}\leqslant\big\|\nabla H(x)\big\|,
\end{align*}
for $x\in B_{r(x_{s})}(x_{s})$.
\end{lemma}

\begin{definition}\label{def2}
Let $x(t)$ on $t\in[0,+\infty)$, be some trajectory of \eqref{mg}. For any
$t>0$, the length of the trajectory on $[0,t)$ is given by
\begin{align*}
l_{[0,t)}=\int_{0}^{t}\big\|\dot{x}(s)\big\|ds.
\end{align*}
\end{definition}
It was pointed out in \cite{Tpc} that finite length implied the
convergence of the trajectory, and was also used to discuss the
global stability of the analytic neural networks in \cite{Mfa}.

{\color{blue}\section{Distributed event-triggered design}\label{sec3}
In this section we synthesize distributed triggers that prescribe when neurons should broadcast state information and update their control signals.
Section \ref{Primary} presents the evolution of a quadratic function that measures network disagreement to identify a triggering function and discusses the problems that arise in its implementation. These observations are our starting point in Section \ref{MorseSardTheorem} and Section \ref{ZenoBehavior}, where we should overcome these implementation issues.}

To design appropriate triggering time point $\{t^{i}_k\}_{k=0}^{+\infty}$ of system \eqref{mg} for $i=\oneton$, we define the state measurement error vector $e(t)=[e_{1}(t),\cdots,e_{n}(t)]^{\top}$ where
\begin{align*}
e_{i}(t)=\Big[\nabla f\big(y(t)\big)\Big]_{i}-\Big[\nabla f\big(y(t_{k}^{i})\big)\Big]_{i}
\end{align*}
for $t\in[t_k^{i},t_{k+1}^{i})$ with $i=\oneton$ and $k=\zerotoinfty$.

\subsection{Distributed Event-triggered Rule}\label{Primary}

{\color{blue}
To design the triggering function $T_{i}(e_{i},t)$ for the updating rule, we define a function vector $\Psi(t)=[\Psi_{1}(t),\cdots,\Psi_{n}(t)]^{\top}$ with
\begin{align*}
\Psi_{i}(t)=\delta(t)\,{\rm e}^{-d_{i}(t-t_{k}^{i})},
\end{align*}
where \footnote{\color{blue}~The function $\Psi_{i}(t)$ can be thought of as a normalized function of $|F_{i}(x(t))|$ by exponential decay function ${\rm e}^{-d_{i}(t-t_{k}^{i})}$. Thus the coefficient $\delta(t)$ with respect to time $t$ in Eq. \eqref{normalization} can be seen as a parameter from this normalization process.}

\begin{align}\label{normalization}
\delta(t)=\frac
{\sqrt{\sum\limits_{i=1}^{n}\Big|F_{i}\big(x(t)\big)\Big|^{2}}}
{\sqrt{\sum\limits_{i=1}^{n}{\rm e}^{-2d_{i}(t-t_{k}^{i})}}},
\end{align}
for $i=\oneton$ and $k=\zerotoinfty$.

What we can observe is a neuron's state $x_{i}(t)$ at a particular time point or a time period (a subset of $[0,+\infty)$) from system \eqref{mg}. Given a specific analytic function $f(\cdot)$, we can directly figure out the right-hand term $F_{i}(x(t))$ without knowing the theoretical formula of the trajectory $x(t)$ of the system \eqref{mg} on $[0,+\infty)$ in advance. Thus, $\delta(t)$ can also be calculated straightly. The samplings for $x_{i}(t)~(i=\oneton)$ with continuous monitoring or discrete-time monitoring would determine the efficiency level and the adjustment cost for the system's convergence. The continuous monitoring can ensure a high level of efficiency with large costs, while discrete-time monitoring on $x_{i}(t)$ can reduce the cost, but sacrifice the efficiency. In Section \ref{Monitoring}, we will discuss the discrete-time monitoring and give a prediction algorithm for the triggering time point $t_{k}^{i}$ based on the obtained information of $x_{i}(t)$ for all $i=\oneton$.
}

\begin{theorem}\label{PrimaryRule}
Set $t_{k+1}^{i}$ as the time point by the updating rule that
\begin{align*}
t_{k+1}^{i}=\max_{\tau\geqslant t_{k}^{i}}\Big\{\tau:T_{i}\big(e_{i},t\big)
\leqslant 0,~\forall~ t\in[t^i_{k},\tau)\Big\},
\end{align*}
that is
\begin{align}\label{PrimaryRule1}
t_{k+1}^{i}=\max_{\tau\geqslant t_{k}^{i}}\bigg\{\tau&:\big|e_{i}(t)\big|
\leqslant\gamma\Psi_{i}(t),~\forall~ t\in[t^i_{k},\tau)\bigg\},
\end{align}
for $i=\oneton$ and $k=\zerotoinfty$. Then, system \eqref{mg} is convergent.
\end{theorem}

The proof of this theorem comprises of the five propositions as follow.

\begin{proposition}\label{Proposition1}
Under the assumptions in Theorem \ref{PrimaryRule}, $L(x)$ in \eqref{Ly} serves as a strict Lyapunov function for the system \eqref{mg}.
\end{proposition}
\begin{proof} The partial derivative of the candidate Lyapunov function $L(x)$ along the trajectory $x(t)$ can be written as \footnote{~In order to avoid ambiguity, we point out that $F_i(x(t))=-d_{i}x(t)-\big[\nabla f\big(y(t^i_{k_i(t)})\big)\big]_{i}+\theta_{i}$, for $i=\oneton$. }
\begin{align}
&\frac{\partial}{\partial x_i}L\big(x(t)\big)\nonumber\\
=&-\lambda_{i}g'\big(\lambda_{i}x_{i}(t)\big)\bigg\{-d_{i}x_{i}(t)
-\Big[\nabla f\big(y(t)\big)\Big]_{i}+\theta_{i}\bigg\}\nonumber\\
=&-\lambda_{i}g'\big(\lambda_{i}x_{i}(t)\big)\bigg\{-d_{i} x_{i}(t)
-\Big[\nabla f\big(y(t^{i}_{k})\big)\Big]_{i}+\theta_{i}\nonumber\\
 &-\Big[\nabla f\big(y(t)\big)\Big]_{i}
+\Big[\nabla f\big(y(t^{i}_{k})\big)\Big]_{i}\bigg\}\nonumber\\
=&-\lambda_{i}g'\big(\lambda_{i}x_{i}(t)\big)\Big[F_{i}\big(x(t)\big)-e_i(t)\Big],
\label{dLyx}
\end{align}
and the time derivative of $L(x(t))$
\begin{align*}
 &~\frac{\rm d}{{\rm d}t}L\big(x(t)\big)\nonumber\\
=&~\sum_{i=1}^{n}\frac{\partial}{\partial x_{i}}L\big(x(t)\big)\frac{{\rm d}x_{i}(t)}{{\rm d}t}\nonumber\\
=&-\sum_{i=1}^{n}\lambda_{i}g'\big(\lambda_{i}x_{i}(t)\big)\Big[F_{i}\big(x(t)\big)-e_{i}(t)\Big]F_{i}\big(x(t)\big).
\end{align*}
Consider the inequality
\begin{align*}
\Big|e_{i}(t)F_{i}\big(x(t)\big)\Big|\leqslant\frac{1}{2c}\big|e_{i}(t)\big|^{2}+\frac{c}{2}\Big|F_{i}\big(x(t)\big)\Big|^{2}.
\end{align*}
The time derivative $\frac{\rm d}{{\rm d}t}L\big(x(t)\big)=\dot{L}\big(x(t)\big)$ can be bounded as
\begin{align*}
 \dot{L}\big(x(t)\big)
=&-\sum_{i=1}^{n}\lambda_{i}g'\big(\lambda_{i}x_{i}(t)\big)
\bigg[\Big|F_{i}\big(x(t)\big)\Big|^{2}-e_{i}(t)F_{i}\big(x(t)\big)\bigg]\nonumber\\
\leqslant
 &-\Big(1-\frac{a}{2}\Big)\sum_{i=1}^{n}\lambda_{i}g'\big(\lambda_{i}x_{i}(t)\big)\Big|F_{i}\big(x(t)\big)\Big|^{2}\nonumber\\
 &+\frac{1}{2a}\sum_{i=1}^{n}\lambda_{i}g'\big(\lambda_{i}x_{i}(t)\big)\big|e_{i}(t)\big|^{2}\nonumber\\
\leqslant
 &-\alpha\sum_{i=1}^{n}\Big|F_{i}\big(x(t)\big)\Big|^{2}+\beta\sum_{i=1}^{n}\big|e_{i}(t)\big|^{2}
\end{align*}
By using the rule \eqref{PrimaryRule1}, it holds
\begin{align}\label{dLy0.1}
\dot{L}\big(x(t)\big)
\leqslant
 &-\alpha\sum_{i=1}^{n}\Big|F_{i}\big(x(t)\big)\Big|^{2}+\beta\gamma^{2}\sum_{i=1}^{n}\Psi_{i}^{2}(t)\nonumber\\
=&-\alpha\sum_{i=1}^{n}\Big|F_{i}\big(x(t)\big)\Big|^{2}
  +\beta\gamma^{2}\sum_{i=1}^{n}\Big(\delta(t)\Big)^{2}{\rm e}^{-2d_{i}(t-t_{k}^{i})}\nonumber\\
=&-\big(\alpha-\beta\gamma^{2}\big)\sum_{i=1}^{n}\Big|F_{i}\big(x(t)\big)\Big|^{2}\leqslant0
\end{align}
for all $k=\zerotoinfty$. For any $x\notin \mathcal S$, there exits $i_0\in\{\oneton\}$ such that
$F_{i_0}(x)\neq0$. Thus $\dot{L}(x)<0$. Proposition \ref{Proposition1} is proved.
\end{proof}

With the Lyapunov function $L(x)$ for system (\ref{mg}) and the event triggering condition \eqref{PrimaryRule1}, the consequent proof follows \cite{Mfa} with necessary modifications.
\begin{proposition}\label{Proposition2}
There exist finite different energy levels $L_j~(j=\onetom)$, such that each set
of equilibrium points
\begin{align*}
\mathcal S_j=\Big\{x\in\mathcal S:L(x)=L_j\text{ and }j=\onetom\Big\}
\end{align*}
is not empty.
\end{proposition}
\begin{proof}
Given an analytic function $f(\cdot)$, a sigmoid function $g_{i}(\cdot)$ and three constants $d_{i}$, $\theta_{i}$ and $\lambda_{i}$ specifically, it follows that the candidate Lyapunov function $L(x)$ in \eqref{Ly} is analytic on $\R^{n}$.

Suppose that there exist infinite different values
$L_{j}~(j=1,\cdots,+\infty)$ such that $\mathcal S_{j}=\{x\in\mathcal
S:L(x)=L_{j}\}$ is not empty. From Lemma \ref{nontrivial}, it is known that
there exists $r_{1}>0$ such that outside $B_{r_{1}}(\bf 0)$ there are no
equilibrium points. Hence $\mathcal S_{j}\subset B_{r_{1}}(\bf 0)$ for
$j=1,\cdots,+\infty$.

Consider points $x^{j}\in\mathcal S_{j}$ for $j=1,\cdots,+\infty$. Since
$x^{j}\in\mathcal S$, it holds $F(x^{j})=0$ and from Eq. \eqref{dLyx}, $\nabla
L(x^{j})=\bf 0$. Since $\overline{B_{r_{1}}(\bf 0)}$ is a compact set, hence, there
exist a point $\widetilde{x}$ and a subsequence
$\{x^{j_{h}}\}_{h=1}^{+\infty}$ such that $x^{j_{h}}\neq\widetilde{x}$ for
all $h=1,\cdots,+\infty$ and $x^{j_{h}}\to\widetilde{x}$ as $h\to+\infty$.
Since $\nabla L$ is continuous, taking into account that $\nabla
L(x^{j_{h}})=\bf 0$ for all $h=1,\cdots,+\infty$, it results $\nabla
L(\widetilde{x})=\bf 0$.

According to Lemma \ref{Loj}, there exist $r(\widetilde{x})>0$ and
$v(\widetilde{x})\in(0,1)$ such that
$|L(x)-L(\widetilde{x})|^{v(\widetilde{x})}\leqslant\|\nabla
L(x)\|$ for $x\in B_{r(\widetilde{x})}(\widetilde{x})$. Since
$x^{j_{h}}\to\widetilde{x}$ as $h\to+\infty$ and $x^{j_{h}}\in\mathcal
S_{j_{h}}$ have different energy levels $L_{j_{h}}$, we can pick a point
$x^{j_{h_{0}}}\in B_{r(\widetilde{x})}(\widetilde{x})$ such that
$L(x^{j_{h_{0}}})\neq L(\widetilde{x})$. Then
\begin{align*}
0<\big|L(x^{j_{h_{0}}})-L(\widetilde{x})\big|^{v(\widetilde{x})}
\leqslant\big\|\nabla L(x^{j_{h_{0}}})\big\|=0,
\end{align*}
which is a contradiction. This completes the proof.
\end{proof}

Without loss of generality, assume that the energy levels
$L_{j}~(j=\onetom)$ are ordered as $L_{1}>L_{2}>\cdots>L_{m}$. Thus
there exists $\gamma>0$ such that $L_{j}>L_{j+1}+2\gamma$, for any
$j=1,2,\cdots,m-1$. For any given $\varepsilon>0$, define
\begin{align*}
\Gamma_{j}=\Big\{x\in\R^{n}:\text{dist}\,(x,\mathcal S_{j})\leqslant\varepsilon\Big\},
\end{align*}
and
\begin{align}\label{Kset}
\mathcal
K_{j}=\overline{\Gamma}_{j}\bigcap\Big\{x\in\R^{n}:L(x)\in\big[L_{j}-\gamma,L_{j}+\gamma\big]\Big\}.
\end{align}

\begin{proposition}\label{Proposition3}
For $j=1,2,\cdots,m$, $\mathcal K_{j}$ is a
compact set and $\mathcal K_{j}\bigcap\mathcal S=\mathcal S_{j}$.
\end{proposition}
\begin{proof} From Lemma \ref{nontrivial}, $\mathcal S_{j}\in B_{r_{1}}(\bf
0)$ is bounded, hence $\overline{\Gamma}_{j}$ is a compact set and
$\big\{x\in\R^{n}:L(x)\in[L_{j}-\gamma,L_{j}+\gamma]\big\}$ is a closed
set. Thus, $\mathcal
K_{j}=\overline{\Gamma}_{j}\bigcap\big\{x\in\R^{n}:L(x)\in[L_{j}-\gamma,L_{j}+\gamma]\big\}$
is a compact set. Then proterty $\mathcal K_{j}\bigcap\mathcal S=\mathcal
S_{j}$ is an immediate consequence of Proposition \ref{Proposition2}.
\end{proof}

\begin{proposition}\label{Proposition4}
For any trajectory $x(t)$ of the system \eqref{mg} and any given time point $\tau\geqslant0$, let $\mathcal K_{j}$, for some
$j\in\{\onetom\}$, be a compact set as defined in \eqref{Kset}. Then there exist a constant $c_{j}>0$ and an exponent
$v_{j}\in(0,1)$ such that
\begin{align*}
\frac{\Big|\dot{L}\big(x(\tau)\big)\Big|}{\Big\|F\big(x(\tau)\big)\Big\|}
\geqslant c_{j}\Big|L\big(x(\tau)\big)-L_{j}\Big|^{v_{j}},
\end{align*}
for $x(\tau)\in\mathcal K_{j}\setminus \mathcal S$.
\end{proposition}

\begin{proof}
Since the notion $t_{k_{i}(\tau)}^{i}$ is simplified as $t_{k}^{i}$ where $k_{i}(\tau)=\arg\max_{k'}\{t^{i}_{k'}\leqslant\tau\}$, the following equation
\begin{align*}
F_{i}\big(x(\tau)\big)=-d_{i}x_{i}\big(t^{i}_{k_{i}(\tau)}\big)-\bigg[\nabla
f\Big(y\big(t^{i}_{k_{i}(\tau)}\big)\Big)\bigg]_{i}+\theta_{i},
\end{align*}
can be rewritten as
\begin{align*}
F_{i}\big(x(\tau)\big)=-d_{i}x_{i}(t_{k}^{i})-\Big[\nabla f\big(y(t_{k}^{i})\big)\Big]_{i}+\theta_{i},
\end{align*}
for $i=\oneton$. From Eq. \eqref{dLyx} and the condition \eqref{PrimaryRule1}, we have
\begin{align*}
 &~\Big\|\nabla L\big(x(\tau)\big)\Big\|^{2}\nonumber\\
=&\sum_{i=1}^{n}\bigg|\frac{\partial}{\partial x_i}L\big(x(\tau)\big)\bigg|^{2}\nonumber\\
=&\sum_{i=1}^{n}\bigg|\lambda_{i}g'\big(\lambda_{i}x_{i}(\tau)\big)\Big[F_{i}\big(x(\tau)\big)-e_{i}(\tau)\Big]\bigg|^{2}\nonumber\\[2pt]
\leqslant&
  ~\beta_{j}^{2}\sum_{i=1}^{n}
  \Bigg[
   F_{i}^{2}\big(x(\tau)\big)
  +e_{i}^{2}\big(x(\tau)\big)
  +2\Big|F_{i}\big(x(\tau)\big)e_{i}\big(x(\tau)\big)\Big|
  \Bigg]\nonumber\\[2pt]
\leqslant&
  ~\beta_{j}^{2}\sum_{i=1}^{n}
  \Bigg[(1+c)F_{i}^{2}\big(x(\tau)\big)+\bigg(1+\frac{1}{c}\bigg)e_{i}^{2}\big(x(\tau)\big)\Bigg]\nonumber\\[2pt]
\leqslant&
  ~\beta_{j}^{2}(1+c)\sum_{i=1}^{n}\Big|F_{i}\big(x(\tau)\big)\Big|^{2}+
  ~\beta_{j}^{2}\bigg(1+\frac{1}{c}\bigg)\gamma^{2}\sum_{i=1}^{n}\Psi_{i}^{2}(t)\nonumber\\[2pt]
=&~\beta_{j}^{2}(1+c)\Bigg[
  \sum_{i=1}^{n}\Big|F_{i}\big(x(\tau)\big)\Big|^{2}
 +\frac{\gamma^{2}}{c}\sum_{i=1}^{n}\Big(\delta(t)\Big)^{2}{\rm e}^{-2d_{i}(t-t_{k}^{i})}\Bigg]\nonumber\\[2pt]
=&~\beta_{j}^{2}(1+c)\bigg(1+\frac{\gamma^{2}}{c}\bigg)\sum_{i=1}^{n}\Big|F_{i}\big(x(\tau)\big)\Big|^{2}\nonumber\\
=&~\beta_{j}^{2}(1+c)\bigg(1+\frac{\gamma^{2}}{c}\bigg)\Big\|F\big(x(\tau)\big)\Big\|^{2},
\end{align*}
where $\beta_{j}=\max_{i\in\{\oneton\}}\max_{x(\tau)\in\mathcal K_{j}}\{\lambda_{i}g'_{i}(\lambda_{i}x_{i}(\tau))\}$.
Then it holds
\begin{align*}
\Big\|F\big(x(\tau)\big)\Big\|\geqslant h_{j}\Big\|\nabla L\big(x(\tau)\big)\Big\|.
\end{align*}
where
\begin{align*}
h_{j}=\frac{1}{\beta_{j}\sqrt{(1+c)\big(1+\frac{\gamma^{2}}{c}\big)}}.
\end{align*}
From Eq. \eqref{dLy0.1}, we have
\begin{align*}
\Big|\dot{L}\big(x(\tau)\big)\Big|\geqslant\big(\alpha-\beta\gamma^{2}\big)\Big\|F\big(x(\tau)\big)\Big\|^{2}.
\end{align*}

For the point $x(\tau)\in\mathcal K_{j}\setminus \mathcal S$, from Eq. \eqref{dLyx},
$\nabla L(x(\tau))\neq0$. There exists
$r(x(\tau))>0$, $c(x(\tau))>0$ and an exponent
$v(x(\tau))\in(0,1)$ such that
\begin{align*}
\Big\|\nabla L\big(x(\tau)\big)\Big\|\geqslant c\big(x(\tau)\big)\Big|L\big(x(\tau)\big)-L_{j}\Big|^{v(x(\tau))},
\end{align*}
for $x\in B_{r(x(\tau))}(x(\tau))$. Indeed, if
$r(x(\tau))>0$ is small, we have $\nabla L(x)\neq0$ for
$x\in\overline{B_{r(x(\tau))}(x(\tau))}$. Therefore, it holds
\begin{align*}
\frac{\Big|\dot{L}\big(x(\tau)\big)\Big|}{\Big\|F\big(x(\tau)\big)\Big\|}
&\geqslant\big(\alpha-\beta\gamma^{2}\big)\Big\|F\big(x(\tau)\big)\Big\|\\
&\geqslant\big(\alpha-\beta\gamma^{2}\big)h_{j}\Big\|\nabla L\big(x(\tau)\big)\Big\|\\[3pt]
&\geqslant\big(\alpha-\beta\gamma^{2}\big)h_{j}c\big(x(\tau)\big)\Big|L\big(x(\tau)\big)-L_{j}\Big|^{v(x(\tau))}\\[3pt]
&\geqslant~c_{j}\Big|L\big(x(\tau)\big)-L_{j}\Big|^{v_{j}},
\end{align*}
where
\begin{align*}
c_{j}=(\alpha-\beta\gamma^{2})\,h_{j}\min_{x(\tau)\in\mathcal K_{j}}\Big\{c\big(x(\tau)\big)\Big\}
\end{align*}
and
\begin{align*}
v_{j}=\min_{x(\tau)\in\mathcal K_{j}}\Big\{v\big(x(\tau)\big)\Big\}
\end{align*}
for $x(\tau)\in \mathcal K_{i}\setminus \mathcal S$.
\end{proof}

Now, we are at the stage to prove that the length of $x(t)$ on $[0,+\infty)$ is finite. The statement proposition is given as follow.
\begin{proposition}\label{Proposition5}
Any trajectory $x(t)$ of the systm \eqref{mg} has a finite length on $[0,+\infty)$, i.e.,
\begin{align*}
l_{[0,+\infty)} =\int_{0}^{+\infty}\big\|\dot{x}(s)\big\|ds
=\lim_{t\rightarrow+\infty}\int_{0}^{t}\big\|\dot{x}(s)\big\|ds<+\infty.\\[-10pt]
\end{align*}
\end{proposition}

\begin{proof}
Assume without loss of generality that $x(0)$ is not an equilibrium point
of Eq. \eqref{mg}. Due to the uniqueness of solutions, we have
$\dot{x}(t)=F(x(t))\neq 0$ for $t\geqslant 0$, i.e.,
$x(t)\in\R^{n}\setminus \mathcal S$ for $t\geqslant 0$. From Proposition \ref{Proposition1}, it is seen that $L(x(t))$ satisfies
$\dot{L}(x(t))<0$ for $t\geqslant0$, i.e., $L(x(t))$
strictly decreases for $t\geqslant0$. Thus, since $x(t)$ is bounded on
$[0,+\infty)$ and $L(x(t))$ is continuous, $L(x(t))$ will
tend to a finite value $L(+\infty)=\lim_{t\to+\infty}L(x(t))$. From
Proposition \ref{Proposition1} and the LaSalle invariance principle \cite{Mhi},
\cite{Jkh}, it also follows that $x(t)\to\mathcal S~(t\to+\infty)$. Thus,
from the continuity of $L$, it results $L(+\infty)=L_{j}$ for some
$j\in\{\onetom\}$ and $x(t)\to \mathcal S_{j}~(t\to+\infty)$.

Since $x(t)\to \mathcal S_{j}~(t\to+\infty)$ and $L(x(t))\to
L_{j}~(t\to+\infty)$, it follows that there exists $\widetilde{t}>0$ such
that $x(t)\in\mathcal K_{i}$ for $t\geqslant\widetilde{t}$.
By using Proposition \ref{Proposition4}, considering that
$x(t)\in\R^{n}\setminus\mathcal S$ for $t\geqslant 0$ and $x(t)\in\mathcal
K_{i}$ for $t\geqslant\widetilde{t}$, we have that there exists $c_{j}>0$
and $v_{j}\in(0,1)$ such that
\begin{align*}
\frac{\Big|\dot{L}\big(x(t)\big)\Big|}{\Big\|F\big(x(t)\big)\Big\|}
=\frac{-\dot{L}\big(x(t)\big)}{\Big\|F\big(x(t)\big)\Big\|}
\geqslant c_{j}\Big|L\big(x(t)\big)-L(+\infty)\Big|^{v_{j}},
\end{align*}
for $t\geqslant\widetilde{t}$. Then
\begin{align*}
\int_{\widetilde{t}}^{t}\big\|\dot{x}(s)\big\|ds
&=\int_{\widetilde{t}}^{t}\Big\|F\big(x(s)\big)\Big\|ds\\
&\leqslant\frac{1}{c_{j}}
\int_{\widetilde{t}}^{t}\frac{-\dot{L}\big(x(s)\big)}{\Big|L\big(x(s)\big)-L(+\infty)\Big|^{v_{j}}}ds.
\end{align*}
The change of variable $\sigma=L(x(s))$ derives
\begin{align*}
\int_{\widetilde{t}}^{t}\big\|\dot{x}(s)\big\|ds
\leqslant&\,\frac{1}{c_{j}}\int_{L(x(\widetilde{t}))}^{L(x(t))}-\frac{1}{\big|\sigma-L(+\infty)\big|^{v_{j}}}d\sigma\\
=&\,\frac{1}{c_{j}(1-v_{j})}\Bigg\{\Big[L\big(x(\widetilde{t})\big)-L(+\infty)\Big]^{1-v_{j}}\\
 &-\Big[L\big(x(t)\big)-L(+\infty)\Big]^{1-v_{j}}\Bigg\}\\
\leqslant&\,\frac{1}{c_{j}(1-v_{j})}\Big[L\big(x(\widetilde{t})\big)-L(+\infty)\Big]^{1-v_{j}},
\end{align*}
for $t\geqslant\widetilde{t}$. Therefore, we have
\begin{align*}
l_{[0,+\infty)}
&=\int_{0}^{+\infty}\big\|\dot{x}(s)\big\|ds\\
&\leqslant\int_{0}^{\widetilde{t}}\big\|\dot{x}(s)\big\|ds+\int_{\widetilde{t}}^{+\infty}\big\|\dot{x}(s)\big\|ds\\
&\leqslant\int_{0}^{\widetilde{t}}\big\|\dot{x}(s)\big\|ds+\frac{\Big[L\big(x(\widetilde{t})\big)-L(+\infty)\Big]^{1-v_{j}}}{c_{j}(1-v_{j})}\\
&<+\infty.
\end{align*}
This completes the proof of Proposition \ref{Proposition5}.
\end{proof}

In what follows it remains to address the proof of Theorem \ref{PrimaryRule}, which is given in Section \ref{Primary}.

\textit{Proof of Theorem \ref{PrimaryRule}:}
Suppose that the condition \eqref{PrimaryRule1} holds. Then from Proposition \ref{Proposition5}, for any trajectory $x(t)$ of the system \eqref{mg}, we have
\begin{align*}
l_{[0,+\infty)}
=\int_{0}^{+\infty}\big\|\dot{x}(s)\big\|\,{\rm d}s
=\lim_{t\rightarrow+\infty}\int_{0}^{t}\big\|\dot{x}(s)\big\|\,{\rm d}s<+\infty.
\end{align*}
From Cauchy criterion on limit existence, for any $\varepsilon>0$,
there exists $T(\varepsilon)$ such that when $t_{2}>t_{1}>T(\varepsilon)$, it results
$\int_{t_{1}}^{t_{2}}\big\|\dot{x}(s)\big\|ds<\varepsilon$. Thus,
\begin{align*}
\Big\|x(t_{1})-x(t_{2})\Big\|
=\bigg\|\int_{t_{1}}^{t_{2}}\dot{x}(s)\,{\rm d}s\bigg\|
\leqslant\int_{t_{1}}^{t_{2}}\big\|\dot{x}(s)\big\|\,{\rm d}s
<\varepsilon.
\end{align*}
It follows that there exists an equilibrium point $x^{\bm*}$ of \eqref{mg}, such that $\lim_{t\rightarrow+\infty}x(t)=x^{\bm*}$. Recalling the Definition \ref{convergence}, we can obtain that system \eqref{mg} is convergence.\QED

\begin{remark}
The event-triggered condition \eqref{PrimaryRule1} implies that the next time interval for neuron $v_{i}$ depends on states of the neurons $v_{j}$ that are synaptically linked to neuron $v_{i}$
We say that neuron $v_{j}$ is synaptically linked to neuron $v_{i}$ if $[\nabla f(y)]_{i}$ depends on $y_{j}$, in other words,
\begin{align*}
\frac{\partial^{2} f(y)}{\partial y_{i}\partial y_{j}}\ne 0.
\end{align*}
\end{remark}

{\color{blue}
It seems naturally that when the event triggers, the neuron $v_{i}$ has to send its current state information $x_{i}(t)$ to its out-neighbours immediately in order to avoid having $\frac{\rm d}{{\rm d}t}L(x(t))>0$. However, such a trigger would have the following problems:
\begin{enumerate}
\renewcommand{\labelenumi}{(P\arabic{enumi})}
\item\label{P1} The triggering function $T_{i}(e_{i},t)=0$ may hold even after neuron $v_{i}$ sends its new state to its neighbours. A bad situation is that $\Psi(t)=0$ happens at the same time when $|e_{i}(t)|=0$. This may cause the neuron to send its state continuously. This is called {\it continuous triggering situation} in the Zeno behavior \footnote{~Zeno behavior is described as a system making an infinite number of jumps (i.e. triggering events in this paper) in a finite amount of time (i.e. a finite time interval in this paper), see \cite{Joh}.}.
\item\label{P2} Event if $\Psi(t)=0$ and $|e_{i}(t)|=0$ never happen at the same time point. The Zeno behavior may still exist. For example, one neuron $v_{i}$ broadcasting its new state to its out-neighbours may cause the triggering rules for two neurons $v_{j_{1}}$ and $v_{j_{2}}$ in $N_{i}^{\text{out}}$ are broken alternately. That is to say, the inter-event time for both $v_{j_{1}}$ and $v_{j_{2}}$ will decrease to zero. This is called {\it alternate triggering situation} in the Zeno behavior.
\end{enumerate}
These observations motivate us to introduce the Morse-Sard Theorem for avoiding the {\it continuous triggering situation} (P\ref{P1}) in Subsection \ref{MorseSardTheorem}. In Subsection \ref{ZenoBehavior}, we will also prove that for all the neuron $v_{i}~(i=\oneton)$, the {\it alternate triggering situation} is absent by using our distributed event-triggered rule in Theorem \ref{PrimaryRule}.
}

\subsection{Exclusion of Continuous Triggering Situation }\label{MorseSardTheorem}
From the rule \eqref{PrimaryRule1}, we know that a triggering event happens at a threshold time $t_{k}^{i}$ satisfying
\begin{align*}
T_{i}(e_{i},t_{k}^{i})=\big|e_{i}(t_{k}^{i})\big|-\gamma\Psi_{i}(t_{k}^{i})=0
\end{align*}
for $i=\oneton$ and $k=\zerotoinfty$.

To avoid the situation that $\Psi_{i}(t)=0$ and $|e_{i}(t)|=0$ happen at the same triggering time point $t_{k}^{i}$ for some $k$, when the triggering function $T_{i}(e_{i},t)=0$ still holds after the neuron $v_{i}$ sends the new state to its neighbours, we define a function vector
\begin{align*}
S\big(t,t_{\tau}\big)=\frac{1}{2}\Big[e^{\top}(t)e(t)-\gamma^{2}\Psi^{\top}(t)\Psi(t)\Big]
\end{align*}
where
\begin{align}\label{ThresholdTime}
t_{\tau}\in\bigcup_{i=1}^{n}\Big\{t^{i}_k:t_{k}^{i}\leqslant t \text{~and~} k=\zerotoinfty\Big\}
\end{align}
is one of the triggering time points before the present time $t$ and $S(t,t_{\tau})=\big[S_{1}(t,t_{\tau}),\cdots,S_{n}(t,t_{\tau})\big]^{\top}$. The following Morse-Sard theorem will be used for excluding this continuous triggering.

\begin{theorem}[Morse-Sard]\label{MSTheorem}
For each initial data $x(0)$, there exists a measure zero subset $\mathcal O\subset\R^{n}$ such that for any given neuron $v_{i}~(i=\oneton)$, the threshold time $t_{\tau}$ for $S(t_{k+1}^{i},t_{\tau})=0$ corresponding to initial data $x(0)\in\R^{n}\backslash\mathcal O$ are countable for all $k=\zerotoinfty$. That is to say, the triggering time point set
\begin{align*}
\bigcup_{i=1}^{n}\bigcup_{k=0}^{+\infty}\big\{t_{k}^{i}\big\}
\end{align*}
is a countable set.
\end{theorem}

\begin{proof}
To show that the threshold time $t_{\tau}$ are countable for each $x(0)\in\R^{n}\backslash\mathcal O$, we prove a statement that the Jacobian matrix ${\rm d}S(t,t_{\tau})=\big[{\rm d}S_{1}(t,t_{\tau}),\cdots,{\rm d}S_{n}(t,t_{\tau})\big]^{\top}$ has rank $n$ at next triggering time point $t=t_{k+1}^{i}$, where
\begin{align*}
{\rm d}S_{i}\big(t,t_{\tau}\big)=
\bigg[
\frac{\partial}{\partial t}S_{i}\big(t,t_{\tau}\big),
\frac{\partial}{\partial t_{\tau}}S_{i}\big(t,t_{\tau}\big)
\bigg].
\end{align*}
The two components of the above equation satisfy
\begin{align*}
 \frac{\partial}{\partial t}S_{i}\big(t,t_{\tau}\big)
=&~e_{i}(t)\frac{{\rm d}e_{i}(t)}{{\rm d}t}-\gamma^{2}\Psi_{i}(t)\frac{{\rm d}\Psi_{i}(t)}{{\rm d}t}\\
=&~e_{i}(t)\frac{{\rm d}}{{\rm d}t}\Big[\nabla f\big(y(t)\big)\Big]_{i}
  +\gamma^{2}d_{i}\delta_{k}^{2}\,{\rm e}^{-2d_{i}(t-t_{k}^{i})}\\
=&~e_{i}(t)\frac{{\rm d}}{{\rm d}t}\Big[\nabla f\big(y(t)\big)\Big]_{i}
  +\gamma^{2}d_{i}\Psi_{i}^{2}(t)
\end{align*}
and
\begin{align*}
  \frac{\partial}{\partial t_{\tau}}S_{i}\big(t,t_{\tau}\big)
=&~e_{i}(t)\frac{{\rm d}e_{i}(t)}{{\rm d}t_{\tau}}
  -\gamma^{2}\Psi_{i}(t)\frac{{\rm d}\Psi_{i}(t)}{{\rm d}t_{\tau}}\\
=&-e_{i}(t)\frac{{\rm d}}{{\rm d}t_{\tau}}\Big[\nabla f\big(y(t_{\tau})\big)\Big]_{i}
  -\gamma^{2}d_{i}\delta_{k}^{2}\,{\rm e}^{-2d_{i}(t-t_{k}^{i})}\\
=&-e_{i}(t)\frac{{\rm d}}{{\rm d}t_{\tau}}\Big[\nabla f\big(y(t_{\tau})\big)\Big]_{i}
  -\gamma^{2}d_{i}\Psi_{i}^{2}(t)
\end{align*}
When event triggers and $e_{i}(t)$ resets to $0$ in the short time period after the next time point $t_{k+1}^{i}$, that is, $e_{i}(t_{k+1}^{i}+\epsilon)\to0$ when $\varepsilon\to0$, then it follows
\begin{align*}
 \lim_{\varepsilon\to0}
 \frac{\partial}{\partial t}S_{i}\big(t,t_{\tau}\big)\bigg|_{t=t_{k+1}^{i}+\varepsilon}
=\gamma^{2}d_{i}\Psi_{i}^{2}(t_{k+1}^{i})
\end{align*}
and
\begin{align*}
  \lim_{\varepsilon\to0}
  \frac{\partial}{\partial t_{\tau}}S_{i}\big(t,t_{\tau}\big)\bigg|_{t=t_{k+1}^{i}+\varepsilon}
=-\gamma^{2}d_{i}\Psi_{i}^{2}(t_{k+1}^{i}).
\end{align*}
Define a initial data set for neuron $v_{i}$ by
\begin{align*}
\mathcal O_{k}^{i}=\Big\{x_{i}(0)\in\R:\delta_{k}^{j}=0,\text{~for all~}j=\oneton\Big\}
\end{align*}
and it holds $\measure(\mathcal O_{k}^{i})=0$ in the sense of Lebesgue measure. Take the initial data $x_{i}(0)$ from $\R\backslash\mathcal O_{k}^{i}$, we have
\begin{align*}
\delta\big(t_{k+1}^{i}\big)=\frac
{\sqrt{\sum\limits_{j=1}^{n}\Big|F_{j}\big(x(t_{k+1}^{i})\big)\Big|^{2}}}
{\sqrt{\sum\limits_{j=1}^{n}{\rm e}^{-2d_{j}(t_{k+1}^{i}-t_{k}^{j})}}}
~\neq0
\end{align*}
that is,
\begin{align*}
\Psi_{i}\big(t_{k+1}^{i}\big)=\delta\big(t_{k+1}^{i}\big)\,{\rm e}^{-d_{i}(t_{k+1}^{i}-t_{k}^{i})}\neq0
\end{align*}
which implies
\begin{align*}
{\rm d}S_{i}\big(t_{k+1}^{i},t_{\tau}\big)
=&\lim_{\varepsilon\to0}{\rm d}S_{i}\big(t_{k+1}^{i}+\varepsilon,t_{\tau}\big)\\
=&\lim_{\varepsilon\to0}\bigg[
  \frac{\partial}{\partial t}S_{i}\big(t,t_{\tau}\big),
  \frac{\partial}{\partial t_{\tau}}S_{i}\big(t,t_{\tau}\big)
  \bigg]\bigg|_{t=t_{k+1}^{i}+\varepsilon}\\
=&\Big[\gamma^{2}d_{i}\Psi_{i}^{2}(t_{k+1}^{i}),-\gamma^{2}d_{i}\Psi_{i}^{2}(t_{k+1}^{i})\Big]\\[2pt]
\neq&0
\end{align*}
thus, for each initial data $x(0)\in\R^{n}\backslash\mathcal O$ with
\begin{align}\label{ZeroMeasuredSet}
\mathcal O=\bigcup_{k=1}^{\infty}\bigcup_{i=1}^{n}\mathcal O_{k}^{i}
\text{ ~and~ }
\measure(\mathcal O)=0,
\end{align}
the Jacobian matrix ${\rm d}S(t,t_{\tau})$ has rank $n$ at time $t=t_{k+1}^{i}$.

Now using the inverse function theorem at each $x(0)\in\R^{n}\backslash\mathcal O$, we can obtain that
for each threshold time $t_{\tau}$ defined in Eq. \eqref{ThresholdTime},
the next triggering time point $t_{k+1}^{i}$ is isolated, hence the set
\begin{align*}
\bigcup_{i=1}^{n}\bigcup_{k=0}^{+\infty}\big\{t_{k}^{i}\big\}
\end{align*}
is a countable set. The Morse-Sard theorem is proved.

\end{proof}

Recalling the triggering function $T_{i}(e_{i},t)$, we can obtain the results that if the initial data $x(0)\in\R^{n}\backslash\mathcal O$, then
\begin{align*}
\Psi(t_{\tau})\neq0
\text{~ for all ~}
t_{\tau}\in\bigcup_{i=1}^{n}\bigcup_{k=0}^{+\infty}\big\{t_{k}^{i}\big\}
\end{align*}
that is to say, $\Psi_{i}(t)=0$ and $|e_{i}(t)|=0$ may never happen at the same time at all the triggering time point $t_{k}^{i}$ where $i=\oneton$ and $k=\zerotoinfty$. Therefore, the {\it continuous triggering situation} in the Zeno behavior (P\ref{P1}) is avoided.

\begin{remark}
To refrain $x(0)$ from the zero measured subset $\mathcal O$, a small perturbation on initial data $x(0)$ can be introduced, which can make it be away from the value that leads to $x(0)\in\mathcal O$. The small perturbation on initial data has no influence on the convergence of the system, for the equilibria of the system \eqref{mg} do not depend on the initial data  sensitively.
\end{remark}

\subsection{Exclusion of Alternate Triggering Situation}\label{ZenoBehavior}
After we exclude the {\it continuous triggering situation} in the above section, what remains is  the {\it alternate triggering situation} in the Zeno behavior. To prove that this situation is absent when using the distributed event-triggered rule \eqref{PrimaryRule1}, we will find a common positive lower-bound for all the inter-event time $t_{k+1}^{i}-t_{k}^{i}$, where $i=\oneton$ and $k=\zerotoinfty$.
\begin{theorem}\label{Zeno}
Let $\mathcal O$ be a zero measured set as defined in \eqref{ZeroMeasuredSet}. Under the distributed event-triggered rule \eqref{PrimaryRule1} in Theorem \ref{PrimaryRule}, for each $x(0)\in\R^{n}\backslash\mathcal O$, the next inter-event interval of every neuron is strictly positive and has a common positive lower-bound. Furthermore,
the {\it alternate triggering situation} in the Zeno behavior is excluded.
\end{theorem}
\begin{proof}
Let us consider the following derivative of the state measurement error for neuron $v_{i}~(i=\oneton)$
\begin{align*}
  \big|\dot{e}_{i}(t)\big|
&=\Bigg|\sum_{j=1}^{n}\Big[\nabla^2 f\big(y(t)\big)\Big]_{ij}\,\dot{y_{j}}(t)\Bigg|\\[2pt]
&=\Bigg|\sum_{j=1}^{n}\Big[\nabla^2 f\big(y(t)\big)\Big]_{ij}
  \lambda_{j}g'_{j}\big(\lambda_{j}x_{j}(t)\big)F_{j}\big(x(t)\big)\Bigg|\\
&\leqslant
  \Big\|\nabla^{2} f\big(y(t)\big)\Big\|
  \Big\|\Lambda\,\partial g\big(\Lambda x(t)\big)\Big\|
  \sqrt{\sum_{j=1}^{n}\Big|F_{j}\big(x(t)\big)\Big|^{2}}\\
&=\Big\|\nabla^{2} f\big(y(t)\big)\Big\|
  \Big\|\Lambda\,\partial g\big(\Lambda x(t)\big)\Big\|\delta(t)
  \sqrt{\sum_{j=1}^{n}{\rm e}^{-2d_{j}(t-t_{k}^{j})}}\\[2pt]
&\leqslant\mathcal M \big\|\Lambda\big\| \sqrt{n}\,\delta(t),
\end{align*}
where
\begin{align*}
\mathcal M=\max_{t\in[0,+\infty)}\Big\|\nabla^{2} f\big(y(t)\big)\Big\|
\end{align*}
then it follows
\begin{align*}
\bigg|\int_{t_{k}^{i}}^{t}\dot{e}_{i}(s){\rm d}s\bigg|
\leqslant\int_{t_{k}^{i}}^{t}\big|\dot{e}_{i}(s)\big|{\rm d}s
\leqslant \delta(t)\mathcal M \big\|\Lambda\big\|\sqrt{n}\,(t-t_{k}^{i}).
\end{align*}
Based on the distributed event-triggered rule \eqref{PrimaryRule1}, the event will not trigger until $|e_{i}(t)|=\gamma\Psi_{i}(t)$ at time point $t=t_{k+1}^{i}>t_{k}^{i}$. Thus, for each $x(0)\in\R^{n}\backslash\mathcal O$, it holds
\begin{align*}
\gamma\,\delta(t)\,e^{-d_{i}(t_{k+1}^{i}-t_{k}^{i})}
=\big|e_{i}(t)\big|
\leqslant \delta(t)\mathcal M \big\|\Lambda\big\|\sqrt{n}\,\big(t_{k+1}^{i}-t_{k}^{i}\big)
\end{align*}
namely,
\begin{align*}
\frac{\gamma\,e^{-d_{i}\eta_{i}}}{\mathcal M\|\Lambda\|\sqrt{n}}=\eta_{i}
\end{align*}
with $\eta_{i}=t_{k+1}^{i}-t_{k}^{i}$, which possesses a positive solution. Hence, for all the neuron $v_{i}~(i=\oneton)$, the next inter-event time has a common positive lower-bound which follows
\begin{align}\label{LowBound}
\eta=\min_{i\in\{\oneton\}}\bigg\{\eta_{i}:\frac{\gamma\,e^{-d_{i}\eta_{i}}}{\mathcal M\|\Lambda\|\sqrt{n}}=\eta_{i}\bigg\}.
\end{align}

It can be seen that $\eta$ has no concern with all the neurons' states $x_{i}(t)~(i=\oneton)$. Thus, there exists a common positive lower-bound, which is a constant, for the next inter-event interval of each neuron. That is to say, the next triggering time point $t_{k+1}^{i}$ satisfies $t_{k+1}^{i}\geqslant t_{k}^{i}+\eta$ for all $i=\oneton$ and $k=\zerotoinfty$, hence the absence of the {\it alternate triggering situation} in the Zeno behavior (P\ref{P2}) is proved.
\end{proof}

To sum up, we have excluded both the {\it continuous triggering situation} and {\it alternate triggering situation} in the Zeno behavior, when the distributed event-triggered rule is taken into account. Therefore, we can assert that there is no Zeno behavior for all the neurons.

\section{Discrete-time monitoring}\label{Monitoring}
The  continuous monitoring strategy for Theorem \ref{PrimaryRule} may be costly since the state of the system should be observed simultaneously. An alternative method is to predict the triggering time point when inequality \eqref{PrimaryRule1} does not hold and update the triggering time accordingly.

For any neuron $v_{i}~(i=\oneton)$, according to the current event timing $t_{k}^{i}$, its state can be formulated as
\begin{flalign}\label{StateFormula}\raisetag{40pt}
\begin{cases}
x_{i}(t)=x_{i}(t_{k}^{\bm*})+\dfrac{1}{\large d_{i}}
         \bigg\{d_{i}x_{i}(t_{k}^{\bm*})+\Big[\nabla f\big(y(t_{k}^{i})\big)\Big]_{i}-\theta_{i}\bigg\}\\[10pt]
\hspace{7.5ex}\times\Big[e^{-d_{i}(t-t_{k}^{\bm*})}-1\Big]\\[7pt]
y_{i}(t)=g_{i}\big(\lambda_ix_i(t)\big)\\[3pt]
\end{cases}&&
\end{flalign}
for $t_{k}^{\bm*}<t<t_{k+1}^{i}$, where $t_{k}^{\bm*}$ is the newest timing of all $v_{i}$'s in-neighbours, that is
\begin{align*}
t_{k}^{\bm*}=\max_{j\in N_{i}^{\text{in}}}\big\{t_{k}^{j}\big\}
\end{align*}
and $t_{k+1}^{i}$ is the next triggering time point at which neuron $v_{i}$ happens the triggering event. Then, solving the following maximization problem
\begin{align}\label{Maximization}
\Delta t_{k}^{i}
=\max_{t\in(t_{k}^{\bm*},t_{k+1}^{i})}\Big\{t-t_{k}^{\bm*}:\big|e_{i}(t)\big|\leqslant\gamma\Psi_{i}(t)\Big\},
\end{align}
we have the following prediction algorithm (Algorithm \ref{Algorithm}) for the next triggering time point.

With the information of each neuron at time $t^{i}_{k}$ and the proper parameters $\gamma$, search the observation time $\Delta t_{k}^{i}$ by \eqref{Maximization} at first. If no triggering events occur in all $v_{i}$'s in-neighbours during $(t_{k}^{i},t_{k}^{\bm*}+\Delta t_{k}^{i})$, the neuron $v_{i}$ triggers at time $t_{k}^{\bm*}+\Delta t_{k}^{i}$ and record as the next triggering event time $t_{k+1}^{i}$, that is $t_{k+1}^{i}=t_{k}^{\bm*}+\Delta t_{k}^{i}$. Renew the neuron $v_{i}$'s state and send the renewed information to all its out-neighbours. The prediction of neuron $v_{i}$ is finished. If some in-neighbours of $v_{i}$ triggers at time $t\in(t_{k}^{i},t_{k}^{\bm*}+\Delta t_{k}^{i})$, update $t_{k}^{\bm*}$ in state formula \eqref{StateFormula} and go back to find a new observation time $t_{k}^{\bm*}+\Delta t_{k}^{i}$ by solving the maximization problem \eqref{Maximization}.

\begin{algorithm}[h]
\caption{Prediction for the next triggering time point $t_{k+1}^{i}$}
\label{Algorithm}
$~~//$ $N_{i}^{\text{in}}$ is the set of neuron $v_{i}$'s in-neighbours\hspace*{\fill}
\begin{algorithmic}[1]
\Require
\State Initialize $\gamma>0$
\State $t_{k}^{\bm*}\gets \max_{j\in N_{i}^{\text{in}}}t_{k}^{j}$
\State $x_{i}(t)\gets x_{i}(t_{k}^{i})$ for all $i=\oneton$ 
\Ensure
\State $\text{Flag}\gets 0$
\While {Flag $=0$} 
\State Search $\Delta t_{k}^{i}$ by the strategy \eqref{Maximization}
\State $\tau\gets t_{k}^{\bm*}+\Delta t_{k}^{i}$
\If{No in-neighbours of $v_{i}$ trigger during $(t_{k}^{i},\tau)$}
\State $v_{i}$ triggers at time $t_{k+1}^{i}=\tau$
\State $v_{i}$ renew its state information $x_{i}(t_{k+1}^{i})$
\State $v_{i}$ sends the state information to its out-neighbours
\State $\text{Flag}\gets1$ %
\Else
\State Update $t_{k}^{\bm*}$ in the state formula \eqref{StateFormula}
\EndIf
\EndWhile
\State\Return{$t_{k+1}^{i}$}
\end{algorithmic}
\end{algorithm}

In addition, when neuron $v_{i}$ updates its observation time $\Delta t_{k}^{i}$, the triggering time predictions of $v_{i}$'s out-neighbours will be affected. Therefore, besides the state formula \eqref{StateFormula} and the maximization problem \eqref{Maximization} as given before, each neuron should take their triggering event time whenever any of its in-neighbours renews and broadcasts its state information. In other word, if one neuron updates its triggering event time, it is mandatory to inform all its out-neighbours.

\begin{remark}
The discrete-time monitoring by using the state formula \eqref{StateFormula} may lose the high-level efficiency of the convergence, because it abandons the continuous adjustment on $\delta(t)$ as defined in Eq. \eqref{normalization}. But the advantage is that a discrete-time inspection on $x(t)$ can be introduced to ensure the convergence in Theorem \ref{PrimaryRule}. This can reduce the monitoring cost.
\end{remark}

\section{Examples}\label{sec5}
In this section, two numerical examples are given to demonstrate the effectiveness of the presented results and the application.

\noindent{\bf Example 1:} Considering a 2-dimension analytic neural network with
\begin{align*}
f(y)=-\frac{1}{2}y^{\top}Wy,
\end{align*}
where
\begin{align*}
D=\Lambda=I_2,~
\theta=\begin{bmatrix} 1\\-1 \end{bmatrix},~
W=\begin{bmatrix} 2 & -1\\ -1 & 2 \end{bmatrix}.
\end{align*}
We have
\begin{align*}
\nabla f(y)=-Wy
\end{align*}
and adopt the distributed event-triggered rule (Theorem \ref{PrimaryRule}). The initial value of each neuron is randomly selected in the interval $[-1,1]$. Figure \ref{fig:1} shows that the state $x(t)$ converges to $\nu=[2.7072,-1.6021]^{\top}$  with the initial value $x(0)=[-0.6703, 0.6304]^{\top}$  by taking $\gamma=0.5$.

\begin{figure}[h]
\centerline{
\includegraphics[width=0.48\textwidth]{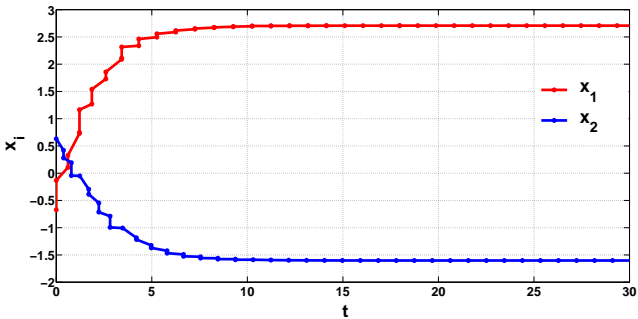}}
\caption{The state $x(t)$ of two neurons converges to $\nu=[2.7072,-1.6021]^{\top}$ with $x(0)=[-0.6703, 0.6304]^{\top}$.}
\label{fig:1}
\end{figure}

Take the different values of the parameter $\gamma$ under the distributed event-triggered rule, the simulation results are shown in Table \ref{Table1}. In this table, $\eta$ is the theoretical lower-bound for the inter-event time of all the neurons calculated by \eqref{LowBound}. $\Delta t_{\min}=\min_{i\in\{\oneton\}}\min_{k\in\{\zerotoinfty\}}(t_{k+1}^{i}-t_{k}^{i})$ is the
actual calculation value of the minimal length of inter-event time. $N$ is number of triggering times and $T_1$ stands for the first time when $\|x(t)-\nu\|\le 0.0001$, as an index for the convergence rate. All results are drawn by averaging over 50 overlaps. It can be seen that the actual calculation minimal inter-event time $\Delta t_{\min}$ is larger than the corresponding theoretical lower-bound $\eta$. This implies that we have excluded the Zeno behavior with the lower-bound $\eta$ of the inter-event time for all the neurons. Moreover, The actual number of event $N$ decrease while $T_1$ increases with $\gamma$ increasing, which is in agreement with the theoretical results.

\begin{table}[h]
\centering
\renewcommand\arraystretch{1.25}\setlength{\tabcolsep}{10pt}\small
\caption{Simulation results with different $\gamma$ under the distributed event-
triggered rule}
\begin{tabular}{c|c|c|c|c}
\hline\rule{0pt}{2.0ex}
$\gamma$ & $\eta$ & $\Delta t_{\min}$ & $N$ & $T_{1}$ \\\hline
0.1  &0.4676 &0.6072   &42.10   &31.6612  \\\hline
0.2  &0.4914 &0.8920   &26.90   &32.5330  \\\hline
0.3  &0.5378 &1.0560   &21.16   &32.8354  \\\hline
0.4  &0.5974 &1.1643   &17.92   &33.0597  \\\hline
0.5  &0.6514 &1.2014   &15.58   &33.0533  \\\hline
0.6  &0.7224 &1.2020   &14.48   &33.2765  \\\hline
0.7  &0.7826 &1.2018   &12.70   &33.4677  \\\hline
0.8  &0.8232 &1.2028   &13.22   &33.4744  \\\hline
0.9  &0.8446 &1.2043   &12.76   &33.4854  \\\hline
\end{tabular}\label{Table1}
\end{table}

According to the definition of Lyapunov (or energy) function \eqref{Ly}, if the input $\|\theta\|$ takes a sufficient small value and $\lambda_i\rightarrow+\infty$ for $i=1,2$, then $L(x)\approx f(y)$.
Thus, as an application of our results, system \eqref{mg} with the distributed event-triggered rule can be utilised to seek the local minimum point of $f(y)$ over $\{0,1\}^2$. Denote
\begin{align*}
\overline{y}(\Lambda)=\lim_{t\rightarrow+\infty}g\big(\Lambda x(t)\big),
\end{align*}
where $x(t)$ is the trajectory of the system \eqref{mg}. Thus $\overline{y}(\Lambda)$ is the local minimum point of $H(y)$ as
\begin{align*}
H(y)=-\frac{1}{2}y^{\top}Wy,
\end{align*}
Figure \ref{fig:2} shows that the terminal limit $\overline{y}(\Lambda)$ converge to a local minimum points $[1,1]^{\top}$ as $\lambda_i\rightarrow+\infty$ for $i=1,2$.

\begin{figure}[h]
\centerline{
\includegraphics[width=0.49\textwidth]{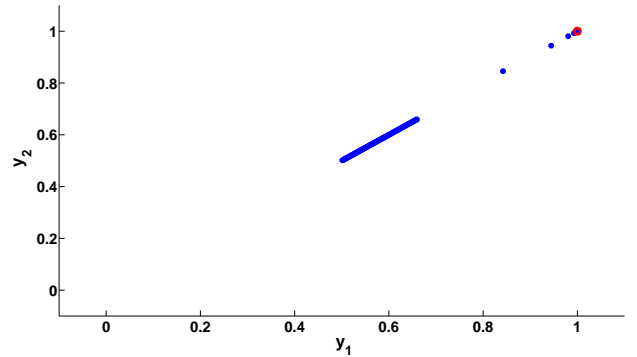}}
\caption{The limit $\overline{y}(\Lambda)$ converges to a local minimum point $[1,1]^{\top}$. We select $\theta=[0.001,0.001]^{\top}$, $\gamma=0.5$, and random initial data in the interval $[-1,1]$. $\lambda_1=\lambda_2$ are selected from 0.01 to 100.}
\label{fig:2}
\end{figure}

\noindent{\bf Example 2:}
Consider a 2-dimension neural network (\ref{mg}) with
\begin{align*}
f(y)=\sum_{i=1}^{2}\bigg(\frac{3}{4}y_{i}^4-y_{i}^3\bigg)-\frac{1}{2}y^{\top}Wy+y
\end{align*}
where
\begin{align*}
D=\Lambda=I_2,~
\theta=\begin{bmatrix} 1\\-1 \end{bmatrix},
W=\begin{bmatrix} 2 & 2\\ 2 & 2 \end{bmatrix}
\end{align*}
We have
\begin{align*}
\nabla f(y)=3\sum_{i=1}^{2}\big(y_{i}^3-y_{i}^2\big)-Wy+1
\end{align*}
and take the distributed event-triggered rule (Theorem \ref{PrimaryRule}). The initial value of each neuron is randomly selected in the interval $[-5,5]$. Figure \ref{fig:3} shows that the state $x(t)$ converges to $\nu=[0.6911,-1.3089]^{\top}$ with taking $\gamma=0.5$.

\begin{figure}[h]
\centerline{\includegraphics[width=0.48\textwidth]{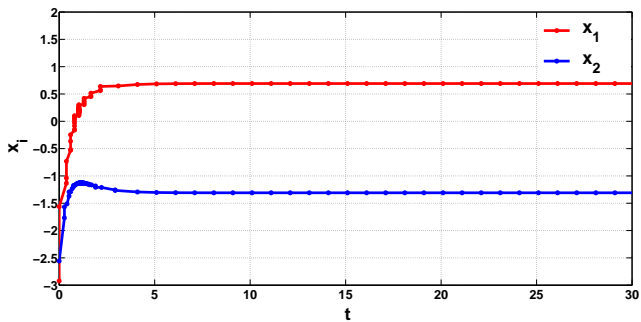}}
\caption{The state $x(t)$ of two neurons converges to $\nu=[0.6911,-1.3089]^{\top}$ with $x(0)=[-2.922, -3.557]^{\top}$.}
\label{fig:3}
\end{figure}

We also calculate the index $\eta$, $\Delta t_{\min}$, $N$ and $T_{1}$ with different values of $\gamma$, as shown in Table \ref{Table2} by averaging over 50 overlaps. The notions are the same as those in Table \ref{Table1}.

\begin{table}[h]
\centering
\renewcommand\arraystretch{1.25}\setlength{\tabcolsep}{10pt}\small
\caption{Simulation results with different $\gamma$ under the distributed event-
triggered rule}
\begin{tabular}{c|c|c|c|c}
\hline\rule{0pt}{2.0ex}
$\gamma$  & $\eta$ & $\Delta t_{\min}$ & $N$ & $T_{1}$ \\\hline
0.1  &0.1540  &0.2892   &20.38   &21.7786  \\\hline
0.2  &0.1639  &0.3052   &19.47   &21.1697  \\\hline
0.3  &0.1716  &0.3308   &18.52   &20.8751  \\\hline
0.4  &0.1854  &0.3654   &17.95   &20.4182  \\\hline
0.5  &0.1983  &0.3939   &18.09   &20.2560  \\\hline
0.6  &0.2014  &0.4213   &17.26   &19.9064   \\\hline
0.7  &0.2154  &0.4582   &16.67   &19.5929   \\\hline
0.8  &0.2279  &0.5106   &16.25   &19.3727   \\\hline
0.9  &0.2348  &0.5352   &16.04   &18.6549   \\\hline
\end{tabular}\label{Table2}
\end{table}

It can also be seen from the Table \ref{Table2} that we have excluded the Zeno behavior with the theoretical lower-bound $\eta$ of the inter-event time smaller than the actual calculation value $\Delta t_{\min}$ under the distributed event-triggered rule. In addition, the actual number of events $N$ decreases while $T_1$ increases with the increasing $\gamma$.

Similar to the first example, if $\|\theta\|$ is sufficiently small and let $\lambda_i\rightarrow+\infty$ for $i=1,2$, it follows $L(x)\approx f(y)$. As an application, we use the distributed event-triggered rule to minimize
\begin{align*}
H(y)=\sum_{i=1}^{2}\bigg(\frac{3}{4}y_{i}^{4}-y_{i}^{3}\bigg)-\frac{1}{2}y^{\top}Wy+y
\end{align*}
over $\{0,1\}^2$. Denote
\begin{align*}
\overline{y}(\Lambda)=\lim_{t\rightarrow+\infty}g\big(\Lambda x(t)\big),
\end{align*}
where $x(t)$ is the trajectory of (\ref{mg}). Then $\overline{y}(\Lambda)$ is the local minimum point of $H(y)$ when $\|\theta\|$ is sufficiently small and $\lambda_i\rightarrow+\infty$. Figure \ref{fig:4} shows that the terminal limit $\overline{y}(\Lambda)$ converges to two local minimum points $[1,0]^{\top}$ and $[0,1]^{\top}$ as $\lambda_i\rightarrow+\infty$ for $i=1,2$.

\begin{figure}[h]
\centerline{
\includegraphics[width=0.49\textwidth]{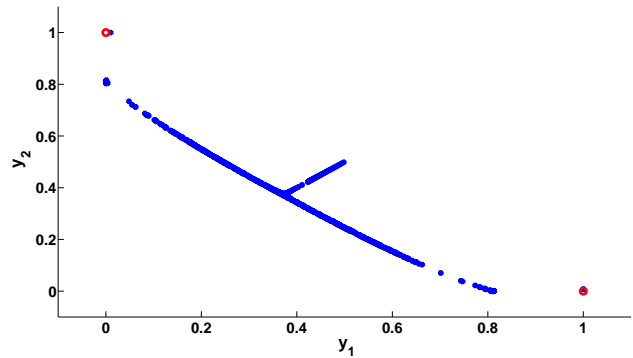}}
\caption{The limit $\overline{y}(\Lambda)$ converges to two local minimum points $[1,0]^{\top}$ and $[0,1]^{\top}$. We select $\theta=[0.001,0.001]^{\top}$, $\gamma=0.5$, and random initial data in the interval $[-5,5]$. $\lambda_1=\lambda_2$ are picked from 0.01 to 100.}
\label{fig:4}
\end{figure}

{
\section{Conclusion}\label{sec6}
In this paper, two triggering rules for discrete-time synaptic feedbacks in a class of analytic neural network {have} been proposed and proved to guarantee neural networks to be completely stable. In addition, the Zeno behaviors can be excluded. By these distributed and asynchronous event-triggering rules, the synaptic information exchanging frequency between neurons are significantly reduced. The main technique of proving complete stability is finite-length of trajectory and the ${\L}$ojasiewicz inequality \cite{Mfa}. Two numerical examples have been provided to demonstrate the effectiveness of the theoretical results. It has also been shown by these examples the application in combinator optimisation, following the routine in \cite{Wll}. Moreever, the proposed approaches can reduce the cost of synaptic interactions between neurons significantly. One step further, our future work will include the self-triggered formulation and event-triggered stability of other more general systems as well as their application in dynamic optimisation.
}

\def\V{\rm vol.~}
\def\N{no.~}
\def\pp{pp.~}
\def\Pot{\it Proc. }
\def\IJCNN{\it International Joint Conference on Neural Networks\rm }
\def\ACC{\it American Control Conference\rm }
\def\SMC{\it IEEE Trans. Systems\rm , \it Man\rm , and \it Cybernetics\rm }

\def\handb{ \it Handbook of Intelligent Control: Neural\rm , \it
    Fuzzy\rm , \it and Adaptive Approaches \rm }


\begin{thebibliography}{22}

\bibitem{Jjh1}
J. J. Hopfield, ``Neural networks and physical systems with emergent collective computational abilities,''
{\it Proc. Nat. Acad. Sci.}, vol. 79, no. 8, pp. 2554-2558,  Apr. 1982.

\bibitem{Jjh}
J. J. Hopfield, ``Neurons with graded response have collective computa-tional properties like those of two-state neurons,''
{\it Proc. Nat. Acad. Sci.}, vol. 81, pp. 3088-3092, May. 1984.

\bibitem{Mvi}
M. Vidyasagar, ``Minimum-seeking properties of analog neural networks with multilinear objective functions,''
{\it IEEE Trans. Automat. Contr.}, vol. 40, pp. 1359-1375, Aug. 1995.

\bibitem{Mfa1}
M. Forti, and A. Tesi, ``New conditions for global stability of neural networks with application to linear and quadratic programming problems,''
{\it IEEE Trans. Circuits Syst. I, Reg. Papers}, vol. 42, no. 7, pp. 354-366, Jul. 1995.

\bibitem{Wll}
W.L. Lu, and J. Wang, ``Convergence analysis of a class of nonsmooth gradient systems,''
{\it IEEE Trans. Circuits Syst. I, Reg. Papers}, vol. 55, no. 11, pp. 3514-3527, Dec. 2008.

\bibitem{Jma}
J. Ma$\acute{n}$dziuk, ''Solving the travelling salesman problem with a Hopfield-type neural network.''
{\it Demonstratio Mathematica}, vol. 29, no. 1, pp. 219-231, 1996.

\bibitem{Mac}
M. A. Cohen, and S. Grossberg, ``Absolute stability of global pattern formation and parallel memory storage by competitive neural networks,''
{\it IEEE Trans. Syst., Man, Cybern.}, vol. 13, no. 15, pp. 815-821, Sep. 1983.

\bibitem{Tpc}
T.P. Chen, and S. Amari, ``Stability of asymmetric Hopfield networks,''
{\it IEEE Trans. Neural Netw.}, vol. 12, no. 1, pp. 159-163, Jan. 2001.

\bibitem{Jdc}
J.D. Cao, and J. Wang, ``Global asymptotic stability of a general class of recurrent neural networks with time-varying delays,'' {\it IEEE Trans. Circuit Syst. I, Fundam. Theory Appl.}, vol. 50, no. 1, pp. 34-44, Jan. 2003.

\bibitem{Wllt}
W.L. Lu, and T.P. Chen, ``New conditions on global stability of Cohen-Grossberg neural networks,''
{\it Neural Comput.}, vol. 15, no. 5, pp. 1173-1189, May 2003.

\bibitem{Tpcl}
T.P. Chen and L.L. Wang, ``Power-rate global stability of dynamical systems with unbounded time-varying delays,''
{\it IEEE Trans. Circuits Syst. II, Exp. Briefs}, vol. 54, no. 8, pp. 705-709, Aug. 2007.

\bibitem{Mfa}
M. Forti, and A. Tesi, ``Absolute stability of analytic neural networks: An approach based on finite trajectory length,''
{\it IEEE Trans. CircuitsSyst. I, Reg. Papers}, vol. 51, no. 12, pp. 2460-2469, Dec. 2004.

\bibitem{Mfp}
M. Forti, P. Nistri, and M. Quincampoix, ``Convergence of neural networks for programming problems via a nonsmooth ${\L}$ojasiewicz inequality,''
{\it IEEE Trans. Neural Netw.}, vol. 17, no. 6, pp. 1471-1486, Nov. 2006.

\bibitem{Mfa2}
M. Forti, and A. Tesi, ``The {\L}ojasiewicz exponent at equilibrium point of a standard CNN is 1/2,''
{\it Int. J. Bifurc. Chaos}, vol. 16, no. 8, pp. 2191-2205, 2006.

\bibitem{Slo}
S. ${\L}$ojasiewicz, ``Une propriet$\acute{e}$ topologique des sous-ensembles analy-tiques r$\acute{e}$els,''
{\it Colloques internationaux du C.N.R.S. Les��quations aux d$\acute{e}$rive$\acute{e}$s partielles}, vol. 117, pp. 87-89, 1963.

\bibitem{Slo1}
S. ${\L}$ojasiewicz, ``Sur la g$\acute{e}$om$\acute{e}$trie semi- et sous-analytique,''
{\it Ann. Inst. Fourier}, vol. 43, pp. 1575-1595, 1993.

\bibitem{Pta}
P. Tabuada, ``Event-triggered real-time scheduling of stabilizing control tasks,''
{\it IEEE Trans. Autom. Control}, vol. 52, no. 9, pp. 1680-1685, Sep. 2007.

\bibitem{Qga}
E. Garcia, P.J. Antsaklis, ``Model-based event-triggered control with time-varying network delays,''
{\it Proc. 50th IEEE Conference on Decision and Control}, pp. 1650-1655, 2011.

\bibitem{Kgv}
K. G. Vamvoudakis, ``An Online Actor/Critic Algorithm for Event-Triggered Optimal Control of Continuous-Time Nonlinear Systems,''
{\it Proc. American Control Conference}, pp. 1-6, Portland, OR, 2014.

\bibitem{AmSh}
A. Molin, S. Hirche, ``Suboptimal Event-Based Control of Linear Systems Over Lossy Channels Estimation and Control of Networked Systems,''
{\it Proc. 2nd IFAC Workshop on Distributed Estimation and Control in Networked Systems}, pp. 5560, 2010.

\bibitem{Mmp}
M. J. Manuel, and P. Tabuada, ``Decentralized event-triggered control over wireless sensor/actuator networks,''
{\it IEEE Trans. Autom. Control}, vol. 56, no. 10, pp. 2456-2461, Oct. 2011.

\bibitem{Xwm}
X.Wang, and M. D. Lemmon, ``Event-triggering distributed networked control systems,''
{\it IEEE Trans. Autom. Control}, vol. 56, no. 3, pp. 586-601, Mar. 2011.

\bibitem{Dvd}
D. V. Dimarogonas, E. Frazzoli, and K. H. Johansson, ``Distributed event-triggered control for multi-agent systems,''
{\it IEEE Trans. Autom. Control}, vol. 57, no. 5, pp. 1291-1297, May 2012.

\bibitem{Zliu}
Z. Liu, Z. Chen, and Z. Yuan, ``Event-triggered average-consensus of multi-agent systems
with weighted and direct topology,''
{\it Journal of Systems Science and Complexity}, vol. 25, no. 5, pp. 845-855, 2012.

\bibitem{Gss}
G. S. Seyboth, D. V. Dimarogonas, and K. H. Johansson, ``Event-based broadcasting for
multi-agent average consensus,''
{\it Automatica}, vol. 49, pp. 245-252, 2013.

\bibitem{Yfg}
Y. Fan, G. Feng, Y. Wang, and C. Song, ``Distributed event-triggered control of multi-agent
systems with combinational measurements,''
{\it Automatica}, vol. 49, pp. 671-675, 2013.

\bibitem{Aapt}
A. Anta, and P. Tabuada, ``Self-triggered stabilization of homogeneous control systems,''
{\it In Proc. Amer. Control Conf.}, 2008, pp. 4129-4134.


\bibitem{Xwmd}
X.Wang, and M. D. Lemmon, ``Self-triggered feedback control systems with finite-gain $\mathcal L_{2}$ stability,''
{\it IEEE Trans. Autom. Control}, vol. 45, no. 3, pp. 452-467, Mar. 2009.

\bibitem{Mjm}
M. J. Manuel, A. Anta, and P. Tabuada, ``An ISS self-triggered implementation of linear controllers,''
{\it Automatica}, vol. 46, no. 8, pp. 1310-1314, 2010.

\bibitem{Aap}
A. Anta, and P. Tabuada, ``To sample or not to sample: self-triggered control for nonlinear systems,''
{\it IEEE Trans. Autom. Control}, vol. 55, no. 9, pp. 2030-2042, Sep. 2010.

\bibitem{Wzh1}
W. Zhu, Z.P. Jian, ``Event-Based Leader-following Consensus of Multi-agent Systems with Input Time Delay''
{\it IEEE Tans. Autom. Contr.}, DOI: 10.1109/TAC.2014.2357131.

\bibitem{Wzh2}
W. Zhu, Z.P. Jian, G. Feng, ``Event-based consensus of multi-agent systems with general linear models''
{\it Automatica}, vol. 50, pp. 552-558, 2014.

\bibitem{YfGf}
Y. Fan, G. Feng, Y. Wang, C. Song, ``Distributed event-triggered control of multi-agent systems with combinational measurements''
{\it Automatica}, vol. 49, pp. 671-675, 2013.

\bibitem{Loc}
L. O. Chua, and L. Yang, ``Cellular neural networks: Theory,''
{\it Trans.Circuits Syst.}, vol. 35, no. 10, pp. 1257-1272, Oct. 1988.

\bibitem{Loc1}
L. O. Chua, and L. Yang, ``Cellular neural networks: Application,''
{\it Trans.Circuits Syst.}, vol. 35, no. 10, pp. 1273-1290, Oct. 1988.

\bibitem{Mhi}
M. Hirsch, ``Convergent activation dynamics in continuous time networks,''
{\it Neural Networks}, vol. 2, pp. 331-349, 1989.

\bibitem{Jkh}
J. K. Hale,
{\it Ordinary Differential Equations}, New York: Wiley, 1980.

\bibitem{Joh}
K.H. Johansson, M. Egerstedt, J. Lygeros, and S.S. Sastry, ``On the regularization of zeno hybrid automata,''
{\it Systems and Control Letters}, vol. 38, pp. 141-150, 1999.

\bibitem{Rdi}
R. Diestel,
{\it Graph theory}, New York: Springer-Verlag Heidelberg, 2005.

\bibitem{RaH}
R.A. Horn, C.R. Johnson,
{\it Matrix Analysis}, Cambridge, U.K.: Cambridge Univ. Press, 1987.









\end{thebibliography}
\end{document}